\documentclass{eptcs}
\usepackage{mathtools,mathpartir,amsmath,amssymb}
\usepackage{hyperref}
\usepackage{keyval}
\usepackage{prooftree}
\usepackage[all]{xy}
\usepackage{dsfont}
\usepackage{enumerate}
\usepackage{color}
\usepackage[dvipsnames]{xcolor}
\usepackage{tikz}                   
\usetikzlibrary{shadows,arrows,shapes,automata,positioning,decorations.markings,fit}
\usepgflibrary{arrows.spaced}
\usepackage{wrapfig}
\usepackage{breakurl}
\usepackage[misc,geometry]{ifsym}
\usepackage{orcidlink}

\title{Asynchronous Multiparty Sessions  with Mixed Choice
}

\def\titlerunning{Asynchronous  Multiparty Sessions  with Mixed Choice 
} 

\def\authorrunning{ 
Barbanera \& Dezani-Ciancaglini
}

\author{Franco Barbanera
\thanks{
Partially supported by 
Project  National Center for ``HPC, Big Data e Quantum Computing",  Programma M4C2 – dalla ricerca all’impresa – Investimento 1.3 – Next Generation EU; and by Project
PIA.CE.RI (PIAno di inCEntivi per la RIcerca di Ateneo) UniCT 2024/2026.
}
\institute{
Dipartimento di Matematica e Informatica,\\
Universit\`a di  Catania, Catania, Italy}
\email{franco.barbanera@unict.it}
\and
Mariangiola Dezani-Ciancaglini
\institute{
Dipartimento di Informatica,\\
Universit\`a di Torino, Torino, Italy}
\email{dezani@di.unito.it}
}

\newtheorem{definition}{Definition}[section]
\newtheorem{lemma}[definition]{Lemma}

\newtheorem{theorem}[definition]{Theorem}

\newtheorem{example}[definition]{Example}

\newenvironment{proof}{{\em Proof.}}{
}

\newcommand{\set}[1]{\{#1\}}
\newcommand{\rn}[1]{\textsc{\footnotesize [{#1}]}}
\newcommand{\coDef}{::=^{coind}}
\newcommand{\mse}{\mu}
\newcommand{\waux}[3]{\mathsf{w\text{-}aux}(#1,#2,#3)}
\newcommand{\w}[2]{\mathsf{w}(#1,#2)}
\newcommand{\tg}{tag}
\newcommand{\tge}{t}
\newcommand{\ms}{message}
\newcommand{\Ms}{Message}
\newcommand{\G}{{\sf G}}

\newcommand{\pplus}{\hspace{2pt}\scalebox{0.7}{\raisebox{1mm}{.....}\hspace{-2.6mm}\raisebox{-1.2mm}{\rotatebox{90}{.....}}}\hspace{6pt}}

\newcommand{\End}{\ensuremath{\mathtt{End}}}

\newcommand{\inact}{\ensuremath{\mathbf{0}}}
\newcommand{\PP}{P}
\newcommand{\Q}{Q}

\newcommand{\ptp}[1]{
  \ensuremath{\mathsf{\color{blue}{ #1}}}
}
\newcommand{\msg}[1]{\mathit{\color{BrickRed}{#1}}}

\newcommand{\pc}{\ptp{c}}

\newcommand{\pp}{\ptp{p}}
\newcommand{\q}{\ptp{q}}
\newcommand{\pq}{\ptp{q}}
\newcommand{\pr}{\ptp{r}}
\newcommand{\ps}{\ptp{s}}

\newcommand{\pw}{\ptp{w}}

\newcommand{\PS}{S}
\newcommand{\PQ}{Q}
\newcommand{\PR}{R}

\newcommand{\SLTS}[1]{\xrightarrow{#1}}

\newcommand{\qed}{\hfill $\Box$}

\newcommand{\pP}[2]{\ensuremath{#1\text{\bf [}#2\text{\bf ]}}}
 \newcommand{\parN}{\mathrel{\|}}
  \newcommand{\parG}{\mathrel{\|}}
 
  \newcommand{\ty}[2]{#1\vdash #2}
 
 \newcommand{\NamedCoRule}[5][]{\IInfer[#1]{#2}{ #3 }{#4}{#5}} 
\newcommand {\IInfer} [5] [] {
  \inferrule*[%
    fraction={===}, 
    left={\textsc{#2}},%
    right={$\begin{array}{l} #5 \end{array}$}, 
    #1
  ]%
   {#3}{#4}}
  \newcommand{\NamedRule}[5][]{\Infer[#1]{#2}{ #3 }{#4}{#5}} 
\newcommand {\Infer} [5] [] {
  \inferrule*[%
    left={\textsc{#2}},%
    right={$\begin{array}{l} #5 \end{array}$}, 
    #1
  ]%
   {#3}{#4}}

   \newcommand{\Nt}{{\mathbb{N}}}
   \newcommand{\Nh}{{\mathbb{M}}}

   \newcommand{\plays}[1]{\ensuremath{{\sf plays}(#1)}}
   \newcommand{\players}[1]{\plays{#1}}

     \newcommand{\GG}{{\mathcal G}}
     
     \newcommand{\mypath}{\sigma}

  \def\finex{{\unskip\nobreak\hfil
\penalty50\hskip1em\null\nobreak\hfil{\Large $\diamond$}
\parfillskip=0pt\finalhyphendemerits=0\endgraf}}

\newcommand{\stackred}[1]{\xrightarrow{#1}}
\newcommand{\concat}[2]{\ensuremath{#1\,{\cdot}\,#2}}

\newcommand{\LL}[1]{{\mathcal L}(#1)}
\newcommand{\LLp}[2]{{\mathcal L}_{#1}(#2)}
\newcommand{\cp}[1]{{\mathsf cap}(#1)}

\newcommand{\cml}{\Lambda}

\newcommand{\mes}[3]{{\langle}#1{,}#2{,}#3{\rangle}}
\newcommand{\Msg}{\mathcal{Q}}

\newcommand{\Cline}[1] {\vspace{0.73mm}\centerline{$ #1 $}\vspace{0.73mm}}
\begin{document}

\maketitle

\begin{abstract}
%
We present an asynchronous calculus for multiparty sessions with mixed choice, which extends the Simple MultiParty Session framework in order to support 
nondeterministic choices with both
input and output prefixes. 
Global types -- equipped with a coinductively defined labelled transition system -- form the basis of a type system that exploits the key notion of coherence of communication 
 label sets. Roughly, a coherent set contains either all the communications 
enabled in the session, or all the actions currently exhibited by a participant,  provided that 
at least one input is enabled for each expected sender. 
Our approach to the typing of multiparty sessions is orthogonal to the well-established, projection-based approaches of the MultiParty Session Type framework.
We prove fundamental theorems for typable multiparty sessions, including Subject Reduction and Session Fidelity. These properties imply that typable sessions are both Lock-Free and Orphan-Message-Free. We also investigate the Eventual Reception property for an extension of the type system. 
 Some  examples demonstrate the expressiveness of mixed choice in asynchronous multiparty protocols and the effectiveness of the proposed type discipline.
\end{abstract}


\section{Introduction}

We present an asynchronous calculus for multiparty sessions with mixed choice, 
 reformulating  
the Mixed Choice Multiparty Session (MCMP) calculus,
 as defined in~\cite{PY24}, 
in the Simple MultiParty Session (SMPS) framework introduced in~\cite{DGD22,BDL23}
 and first used under this name in \cite{BDL24}.  Notably, communication in MCMP is synchronous, 
whereas in our calculus it is asynchronous, providing a more accurate representation of actual scenarios involving distributed participants. 
The present calculus supports 
nondeterministic choices  containing both input and output prefixes. 
A typical simple   interaction exploiting the expressiveness of mixed choice is the following one between  a client and a server. 
 Client $\pc$ keeps on sending task requests ($\msg{req}$) to the server $\ps$, to which $\ps$ 
replies with corresponding results ($\msg{res}$). At any time $\ps$ can decide to terminate the interaction with $\pc$ by sending a $\msg{halt}$ \tg, after which both parties stop. So $\ps$ can decide either to receive a $\msg{req}$ or to send a $\msg{halt}$. 
 Without the option of mixed choice, $\ps$ would be forced to inform $\pc$ of their willingness to proceed with a request/result round at the cost of an additional interaction.

 We equip our calculus  with  a type system grounded on  
the key notion of coherence of communication label 
sets. Roughly, a coherent set contains either all the communications
enabled in the session, or all the actions currently exhibited by a participant,  provided that 
at least one input is enabled for each expected sender. Our approach to the typing of multiparty sessions  -- similarly to recent works such as \cite{CFJ26} --  is orthogonal to the well-established, projection-based approaches of the MultiParty Session Type framework~\cite{HYC08,Honda2016}. 
As in~\cite{DGD22,BDL23,BD24}, 
 we do not distinguish between the local-type and process layers, instead considering a single layer of abstract processes. 
Without resorting to the MPST notion of projection, this enables sessions (i.e. parallel compositions of named processes, our participants) to be directly typed by global types, which are equipped with a coinductively defined labelled transition system.
Fundamental theorems for typable multiparty sessions hold. 
Subject Reduction and Session Fidelity show an isomorphism between the reductions of sessions and 
 those  of the 
global types which type  sessions  
in parallel with their queues. These properties imply that typable sessions are both Lock-Free and Orphan-Message-Free. We also investigate the Eventual Reception property for a restriction 
of the type system. 
 Three examples  
demonstrate the expressiveness of mixed choice in asynchronous multiparty protocols and the effectiveness of the proposed type discipline. 
Full proofs and a further example can be found in \cite{BD26full}.

\section{The Asynchronous SMPS Calculus with Mixed Choice}

We  present now  the SMPS asynchronous calculus of multiparty sessions with  mixed choice, 
inspired  mainly 
by~\cite{PY24} and partially by~\cite{BDL23}. 
We assume to have the following denumerable base sets:  \emph{\tg s}  (ranged
over by $\msg{\lambda},\msg{\lambda'},\dots$); \emph{session participants} (ranged over
by  $\pp,\q,\pr, \ps, \ldots$); \emph{indexes} (ranged over by  $i, j, h, k,\dots$);
\emph{ finite sets of indexes} (ranged over by $I, J, H, K, \dots$).

Processes, ranged over by $\PP,\Q,\PR,\PS,\dots$, 
implement the behaviour of participants. They are sums of processes prefixed by actions.
In the following and in later definitions the symbol $\coDef$ expresses  
that the
productions have to be interpreted \emph{coinductively}  and that only \emph{regular} terms are allowed.  
    Accordingly,  we can adopt in proofs the coinduction style
advocated in \cite{KozenS17} which, without any loss of formal rigour,
 promotes readability and conciseness.

\begin{definition}[Processes]\label{p} 
 \begin{enumerate}[i)]
 \item
 {\em  Action  prefixes} are defined by\quad $\pi\ ::=\ \pp?\msg{\lambda}\ \mid\  \pp!\msg{\lambda}$. 
  \item
  {\em Processes} are  
  defined by
  
  \Cline{\PP\coDef\inact\ \mid\ \Sigma_{i\in I}\pi_i.\PP_i}
  
  \noindent
where $I\neq\emptyset$ and finite, and 
  $\pi_{ l } 
  = \q? \msg{\lambda_{ l }}
  $, $\pi_j= \q? \msg{\lambda_j}$ (resp. $\pi_{ l }
  = \q! \msg{\lambda_{ l }}
  $, $\pi_j= \q! \msg{\lambda_j}$) imply 
 $\msg{\lambda_l}\neq\msg{\lambda_j}$, for any $ l,j  
 \in I$ such that $ l\neq j $. 
 \end{enumerate}
\end{definition}

\noindent
 In the above definition, $\Sigma_{i\in I}\pi_i.\PP_i $ stands, as usual, for the summand of
processes $\pi_i.\PP_i$'s.
A $\Sigma_{i\in I}\pi_i.\PP_i $ process represents  the 
nondeterministic choice of one of the actions $\pi_i$, after which the process
continues as $\PP_i$ with $i\in I$. As usual, we assume the summand of processes to be commutative
and associative. 
A prefix $\pi$ can be any {\em input} (i.e. of the form $\pp?\msg{\lambda}$)
or {\em output} (i.e. of the form $\pp!\msg{\lambda}$) action.  
We use $\inact$ to denote the terminated process. 
For the sake of readability, we omit trailing $\inact$'s in processes.
 The condition on  prefixes  
  enforces  pairwise  distinct  prefixes in summands.  

\noindent
Processes are essentially terms of the calculus MCMP (Mixed Choice Multiparty Sessions) 
 as defined in~\cite{PY24}, the only difference being the absence of communicated values and if-then-else processes.  In the SMPS framework, the latter can be abstractly represented by means of the nondeterministic choice.
 
\smallskip
 {\bf Notation}:  We use 
$\pi.\PP \pplus \Q$   as short for either $\pi.\PP + \Q$ or $\pi.\PP$.
 Such a notation enables us to present the  operational semantics  
 in a compact and yet formal way, since
in our processes -- as in~\cite{PY24} -- we cannot have unprefixed $\inact$'s as summands. 
\smallskip

We use queues in order  to formalise a one-to-one asynchronous model of communication.
Instead of explicitly  defining  a queue for each possible sender and receiver, we use a 
single queue and equip the communicated  \tg s 
with their sender and receiver names, so forming
triples that we dub  {\em \ms s}.

\begin{definition}[\Ms s and Queues]
\begin{enumerate}[i)]
\item
 \emph{\Ms s} are triples of the form $\mes\pp{\msg{\lambda}}\q$. 
\item
{\em \Ms\ queues} (queues for short) are defined by: \quad
$\Msg::=\emptyset \mid \mes\pp{\msg{\lambda}}\q\cdot\Msg$.
\end{enumerate}
  \end{definition}
 A \ms\ $\mes\pp{\msg{\lambda}}\q$ denotes that the participant $\pp$ is the sender of  the \tg\ $\msg{\lambda}$  to  the receiver $\q$. 
 Sent \ms s are stored in a queue,  from which  they are subsequently fetched by the receivers.
The order of \ms s in a queue  determines  the order in which they will be read. 
Since order only matters between \ms s with
the same sender and receiver, we  always  consider \ms\ queues modulo the  following  structural equivalence:

\Cline{\Msg\cdot{\mes\pp{\msg{\lambda}}\q}\cdot{{\mes\pr{\msg{\lambda'}}\ps}\cdot{\Msg'}}\equiv
 {\Msg}\cdot{\mes\pr{\msg{\lambda'}}\ps}\cdot{\mes\pp{\msg{\lambda}}\q}\cdot{\Msg'}
  \quad \text{ if}~~\pp\not=\pr~~\text{or}~~\q\not=\ps 
}

\noindent
Note, in particular, that
$\mes\pp{\msg{\lambda}}\q\cdot\mes\q{\msg{\lambda'}}\pp \equiv
\mes\q{\msg{\lambda'}}\pp\cdot\mes\pp{\msg{\lambda}}\q$. 
These two equivalent queues represent a situation in which  participants $\pp$ and $\q$ have sent each other a \tg, but neither has read theirs yet.  
This situation  is indeed possible in concurrent systems with asynchronous communication. 


\smallskip
Multiparty sessions are networks equipped with a queue, where a network is
a parallel composition of  participant-process pairs 
 of the form $\pP{\pp}{\PP}$.
 
\begin{definition}[Networks and Multiparty Sessions] 
\begin{enumerate}[i)]
\item
{\em Networks} are defined by

\Cline{\Nt = \pP{\pp_1}{\PP_1} \parN \cdots \parN \pP{\pp_n}{\PP_n} }

\noindent
where $\pp_j \neq \pp_l $ for $1\leq j,l\leq n$ and $j\neq l$;
\item 
{\em Multiparty sessions} are defined by \quad$\Nh::=\Nt\parN \Msg $.
\end{enumerate}
\end{definition}
\noindent
  We assume the standard structural congruence  $\equiv$  on networks, stating that
parallel composition is commutative and associative and has neutral elements 
$\pP\pp\inact$ for any $\pp$.
We write $\pP\pp\PP\in\Nt$ if $\Nt\equiv\pP\pp\PP\parN\Nt'$ and $\PP\neq\inact$. Moreover, we define $\plays\Nt=\set{\pp\mid \pP\pp\PP\in\Nt}$ and $\plays{\Nt\parN \Msg}=\plays\Nt$.
 A network $\Nt$ is {\em final} whenever $\plays\Nt=\emptyset$.  

\smallskip
 To define the {\em  asynchronous operational semantics} of sessions we use an LTS,
whose transitions are decorated by {\em communication labels}, i.e.  output labels 
of the shape 
$\pp\q!\msg{\lambda}$ and input labels of the shape $\pp\pq?\msg{\lambda}$, ranged over by $\cml$, $\cml',\ldots$. We say that the 
output label $\pp\q!\msg{\lambda}$ corresponds to the input label $\pq\pp?\msg{\lambda}$. 

We denote as follows the   push/pull  
of a \ms\ in/from a queue.   
 \begin{definition}
 \label{def:pushpull}
 We define $\pp\q!\msg{\lambda}(\Msg)=\Msg\cdot\mes \pp {\msg{\lambda}} \pq$ and $\pp\q?\msg{\lambda}(\mes \pq {\msg{\lambda}} \pp\cdot\Msg)=\Msg$.  
 \end{definition}
 Notice that $\cml(\Msg)$ is undefined only when $\cml=\pp\q?\msg{\lambda}$ and
 $\Msg\not\equiv\mes \pq {\msg{\lambda}} \pp\cdot\Msg'$ for  any 
  $\Msg'$.  

\begin{definition}[LTS for Multiparty Sessions]\label{slts}
The {\em   labelled transition   system   (LTS) 
for multiparty sessions}  
  is the closure under structural congruence of the reduction specified by the axioms: 
  
  \Cline{\begin{array}[c]{@{}c@{}}
      \NamedRule{\rn{Out}}{  
      }{
      \pP\pp  {\q!\msg{\lambda}.\PP \pplus \PP' }
      \parN \Nt\parN \Msg\\
     \SLTS{\pp\q!\msg{\lambda}}\qquad
      \pP\pp{\PP}\parN\Nt\parN \pp\q!\msg{\lambda}(\Msg)
      }{}\\[1mm]
       \NamedRule{\rn{In}}{  
      }{
      \pP\pp  {\q?\msg{\lambda}.\PP \pplus \PP' }
      \parN \Nt\parN \Msg\\
     \SLTS{\pp\q?\msg{\lambda}}\qquad
      \pP\pp{\PP}\parN\Nt\parN \pp\q?\msg{\lambda}(\Msg)
      }{} 
    \end{array}
 }
\end{definition}

\noindent
Notice that the implicit condition on  Axiom  
\rn{In} to be applied is ``$\pp\q?\msg(\Msg)$ is defined''.

\begin{example}[Client/Server]\label{ex1}
 {\em The client/server system  
described in the  introduction 
can 
be implemented by

 \Cline{\pP\pc\PP\parN\pP\ps\PQ\parN\emptyset} 
 
 \noindent where $\PP=\ps!\msg{req}.(\ps?\msg{res}.\PP+\ps?\msg{halt}.\ps?\msg{res})$ and $\PQ=\pc?\msg{req}.\pc!\msg{res}.\PQ+\pc!\msg{halt}.\pc?\msg{req}.\pc!\msg{res}$.}\finex
\end{example}

As usual we define (finite) {\em traces} as  sequences of communication labels:
$\sigma :=\epsilon\mid \cml\cdot\sigma$.
When $\sigma ={\cml_1}\cdot\ldots\cdot{\cml_n}$ ($n\geq 1)$
we write $\Nt\parN\Msg\stackred{\sigma}\Nt'\parN\Msg'$ as short for $\Nt\parN\Msg\stackred{\cml_1}\Nt_1\parN\Msg_1\cdots\stackred{\cml_n}\Nt_{n}\parN\Msg_{n} 
 =  \Nt'\parN\Msg'$. We also write $\Nt\parN\Msg\stackred{\sigma}$ whenever there exists $\Nt'\parN\Msg'$ such that $\Nt\parN\Msg\stackred{\sigma}\Nt'\parN\Msg'$.

 We  conclude  
 this section by defining Lock Freedom following~\cite{Padovani14}
and Orphan-message Freedom following~\cite{DY12,DY13}.
Roughly, a session is Lock Free when there is always a continuation enabling a participant to communicate whenever it is willing to do so.
Lock Freedom entails Deadlock Freedom, since it ensures progress for each participant.
Orphan-message Freedom guarantees that in any reachable session if no participant is active, then no message is left in the queue. 

\begin{definition}[Lock Freedom]\label{d:lf}
  A session  $\Nh$ is \emph{lock free} if\/  $\Nh\SLTS{\mypath}\Nh'$ and
  $\pp\in\plays{\Nh'}$ imply 
	 $\Nh'\SLTS{\concat{\mypath'}\Lambda}$ 
	 for some $\mypath'$ and $\Lambda$ such that $\pp\in\plays\Lambda$.
\end{definition}

\begin{definition}[Orphan-message Freedom]\label{d:omf}
 A session  $\Nh$ is \emph{orphan-message free} if\/  $\Nh\SLTS{\mypath}\Nt'\parN\Msg'$ and 
 $\plays{ \Nt' } 
 =\emptyset$ imply $\Msg'=\emptyset$.
\end{definition}
Orphan-message Freedom is a property usually investigated in the setting of
Communicating Finite State Machines (CFSM)~\cite{bz83,cf05}, a classical formalism for protocol specification and verification. 
As a matter of fact, our  sessions  
can be shown to be equivalent to
systems of CFSMs\footnote{\label{fn} The differences are just in the way we formalise communication notions.
We use regular terms as processes instead of automata. We use a single queue, instead of 
a number of directed channels, one for each pair of distinct participants. Moreover,
in CFSM, the input action $\pp\pq?\msg{\lambda}$
is interpreted as the action of $\pq$ reading the tag $\msg{\lambda}$ from the directed channel $\pp\pq$ (i.e.  the channel containing messages with sender $\pp$ and receiver $\pq$). In SMPS, instead, $\pp\pq?\msg{\lambda}$ denotes the action of $\pp$ reading in the queue the tag $\msg{\lambda}$ sent by  $\pq$ to $\pp$.  
 The  interpretation of  outputs  in the two formalisms is the same.  
 In fact, also   in CFSM $\pp\pq!\msg{\lambda}$ represents the insertion of  the tag 
$\msg{\lambda}$   
 in  the   channel $\pp\pq$   containing the messages from  
 $\pp$ to $\pq$.
 }. 
Orphan-message Freedom is weaker than Eventual Reception, the property guaranteeing
than any message pushed in the queue can be eventually pulled from it. We deal with such a 
property in Section~\ref{pr}.

\section{Type Assignment System}

 Global types represent   the overall behaviour of multiparty sessions  by means of communication labels. 
Here we use a notion of global type  that is  similar to the one in  MPST but, following SMPS,
we define them coinductively as possibly infinite regular terms.

\begin{definition}[Global types]\label{def:gt}
{\em Global types} are  
 defined by:

\Cline{\G\coDef\End\ \mid\ \Sigma_{i\in I}\cml_i.\G_i
         }

\noindent
where $I\neq\emptyset$ and finite, and for any $j, l\in I$ such that $j\neq l$, 
$\cml_j= \pp\q!\msg{\lambda_j}$ and $\cml_l= \pp\q!\msg{\lambda_l}$ imply $\msg{\lambda_j}\neq\msg{\lambda_l}$. 
 Likewise  for $\cml_j= \pp\q?\msg{\lambda_j}$ and $\cml_l= \pp\q?\msg{\lambda_l}$. 
\end{definition} 
As usual, trailing $\End$'s will be omitted. 

\medskip

 The {\em players of
   communication labels} are the senders for the outputs and the receivers for the inputs, i.e. we define
   
\Cline{
 \players{\pp\pq!{\msg{\lambda}}}=
    \players{{\pp\pq}?{\msg{\lambda}}}=\set\pp}

\noindent
Moreover, we define
 $\players{\G}$ by 
 
 \Cline{\players{\End}=\emptyset\qquad\players{\cml.\G}=\players{\cml}\cup\players{\G}\qquad\players{\G_1+\G_2}=\plays{\G_1}\cup\players{\G_2}}
 
 \noindent
   By the regularity condition, $\players{\G}$ is finite for any $\G$.

The notion of coherent set of communication labels is crucial for typing sessions with global types.
 This  requires  defining when a participant is satisfied in a session. 
 Informally, $\pp$  is  
satisfied in  $\Nt\parN \Msg$ when there is at least one \ms\ on the queue  from each
of its required senders in the top choice.  
Then participants with top choices only between outputs are always satisfied.  

\begin{definition}[Participant Satisfaction] A participant $\pp$ is {\em satisfied} in  $\Nt\parN \Msg$ 
 whenever
 
 \Cline{
\pP\pp{\Sigma_{i\in I} \pq_i!\msg{\lambda_i}.\PP_i+\Sigma_{j\in J} \pq_j?\msg{\lambda_j}.\PP_j}\in\Nt \text{ and }L=\set{j\in J \mid \Msg\equiv \mes{\pq_j}{\msg{\lambda_j}}{\pp}\cdot \Msg_j}}

\noindent
 imply that 
for all $\pq_j$ with $j\in J$ there is $l\in L$ such that $\pq_j=\pq_l$. 
\end{definition} 

\noindent
 For example $\ps$ is satisfied in $\pP\pc{\ps?\msg{res}.\PP+\ps?\msg{halt}.\ps?\msg{res}}\parN\pP\ps{\pc?\msg{req}.\pc!\msg{res}.\PQ+\pc!\msg{halt}.\pc?\msg{req}.\pc!\msg{res}}\parN\mes\pc{\msg{req}}\ps$, where $\PP$ and $\PQ$ are defined in Example~\ref{ex1}.

  A coherent set of communications comprises either the communications of a satisfied participant, or  all the communications that participants can engage in.

\begin{definition}[Coherent  Set of Communication Labels]\label{cscl}
Let $\LL{\Nt\parN \Msg}= \set{\cml\mid  \Nt\parN \Msg  \SLTS\cml}$ 
 and  \linebreak $\LLp{\pp}{\Nt\parN \Msg}= \set{\cml\in\LL{\Nt\parN \Msg}\mid \plays{\cml}=\set\pp}$.  A set $\set{\cml_i}_{i\in I}$ is {\em coherent} for $\Nt\parN \Msg$ if  $I\neq\emptyset$ and 
\begin{enumerate}
\item\label{cscl1} either 
$\set{\cml_i}_{i\in I}=$ $\LLp{\pp}{\Nt\parN \Msg}$ for a $\pp$ satisfied in  $\Nt\parN \Msg$;
\item\label{cscl2} or
$\set{\cml_i}_{i\in I}=\LL{\Nt\parN \Msg}$.
\end{enumerate}
\end{definition}

\noindent 
 Let $\PP$ and $\PQ$ be defined as in Example~\ref{ex1}, then
the coherent sets of the session $\pP\pc\PP\parN\pP\ps\PQ\parN\emptyset$ are:
\begin{itemize}
\item $\set{\pc\ps!\msg{req}}$ by Item~\ref{cscl1}, since $\pc$ is satisfied in  $\pP\pc\PP\parN\pP\ps\PQ\parN\emptyset$;
\item $\set{\pc\ps!\msg{req}, \ps\pc!\msg{halt}}$ by Item~\ref{cscl2}.
\end{itemize}

\noindent

We propose a type system in which, when applying a rule in a derivation, a coherent set must be taken into account.
When typing a session, having a satisfied participant ensures that all other participants' communications are recorded in all branches of the inferred global type.
However, when this is not possible, as with interdependent communications, the second case in the definition of coherence must be considered. This means that all possible interactions in the session must be accounted for in the typing process.

\begin{definition}[Type System]\label{def:type-system}
The type system is defined by the following axiom and rule, where sessions are considered modulo structural congruence and the double lines recall that the rule has to be interpreted coinductively~\cite[Chapter 21]{pier02}:  

\smallskip

\Cline{\NamedRule{\rn{End}}{}{\ty\End{\pP\pp\inact}  \parN\emptyset }{}{}}

\medskip

\Cline{
	\NamedCoRule{\rn{TComm}}
{\mbox{$\begin{array}{c} \Nt\parN \Msg\SLTS{\cml_i}\Nt_i\parN \Msg_i\qquad  \ty{\G_i}{\Nt_i\parN \Msg_i}  
\qquad \forall i\in I 
\\
\set{\cml_i}_{i\in I}\text{ is coherent for } \Nt\parN \Msg \qquad\plays{\Sigma_{i\in I}\cml_i.\G_i}=\plays\Nt \\[3mm]
\end{array}$}}
{\ty{\Sigma_{i\in I}\cml_i.\G_i}{\Nt\parN \Msg}}
{ }{}
	}
\end{definition}

\noindent
 It can be checked that the  client/server  system  
 implemented in  Example \ref{ex1} 
 can be typed by  
 
 \Cline{\G=\pc\ps!\msg{req}.(\ps\pc?\msg{req}.\ps\pc!\msg{res}.\pc\ps?\msg{res}.\G+\ps\pc!\msg{halt}.\pc\ps?\msg{halt}.\ps\pc?\msg{req}.\ps\pc!\msg{res}.\pc\ps?\msg{res})
  }

\noindent
 by using only sets which are coherent thanks to 
  condition~\ref{cscl1}.  in  
  Definition~\ref{cscl}.
 
 Premises of  Rule \rn{TComm} can be expressed 
 in terms of the shape of the processes and queues in 
the session. The current formulation which uses the LTS of multiparty sessions is  simpler.

Unless the current session in a type derivation requires to take into account a coherent set containing the labels of all the possible reductions (for instance when participants can be obtained one from the other by means of renaming)
coherent sets enable to consider the reduction behaviour of a single participant at a time in type derivations, as in type systems for sessions without mixed choice.

 The coherence of $\set{\cml_i}_{i\in I}$ for $\Nt\parN \Msg$ 
 in Rule \rn{TComm}   is  
needed for Subject Reduction and Lock Freedom. In fact, without this condition, 
 some branches possibly leading to un-typeable sessions could be disregarded in a typing.
For example, we could derive 

\Cline{\ty{\pp\pq!\msg{\lambda}.\pq\pp?\msg{\lambda}.\pr\pp!\msg{\lambda}.\pp\pr?\msg{\lambda}}{\pP\pp{\pq!\msg{\lambda}.\pr?\msg{\lambda}+\pr?\msg{\lambda}.\pq!\msg{\lambda'}}\parN\pP\pq{\pp?\msg{\lambda}}\parN \pP\pr{\pp!\msg{\lambda}}\parN\emptyset}
}

\noindent
 However, the above session reduces to $\pP\pq{\pp?\msg{\lambda}}\parN\mes{\pp}{\msg{\lambda'}}{\pq}$,
which is stuck and cannot be typed,\footnote{
 Note that the same judgment can be derived using the type systems of~\cite{BD24,BDL23}, highlighting the challenges posed by mixed-choice in relation to existing approaches.
}  

The condition ``$\plays{\Sigma_{i\in I}\cml_i.\G_i}=\plays\Nt$'' 
is   also  necessary to get Lock Freedom. 
 Otherwise,  we could derive 
$\ty{\G}{\pP\pp\PP\parN\pP\pq\PQ\parN\pP\pr{\ps!\msg{\lambda}}\parN\emptyset}$ with $\PP=\pq!\msg{\lambda}.\PP$, $\PQ=\pp?\msg{\lambda}.\PQ$
and $\G=\pp\pq!\msg{\lambda}. \pq\pp?\msg{\lambda}.\G$.   
 In particular,
the  (infinite) derivation uses the satisfaction of $\pp$ when the queue is empty  
and the satisfaction of $\pq$ when the queue contains $\mes\pp{\msg\lambda}\pq$. 

When both conditions of Definition~\ref{cscl} are satisfied, we get more  compact  
types  by  using the first  one.   
For example,  by using the first condition when typing  
$\pP\pp{\pq?\msg{\lambda}}\parN\pP\pr{\ps?\msg{\lambda'}}\parN\mes{\pq}{\msg{\lambda}}{\pp}\cdot\mes{\ps}{\msg{\lambda'}}{\pr}$
we can  derive either $\pp\pq?\msg{\lambda}.\pr\ps?\msg{\lambda'}$ or $\pr\ps?\msg{\lambda'}.\pp\pq?\msg{\lambda}$. 
With the second condition we  can  only derive 
$\pp\pq?\msg{\lambda}.\pr\ps?\msg{\lambda'}+\pr\ps?\msg{\lambda'}.\pp\pq?\msg{\lambda}$.

\smallskip
We define  now  the semantics of global types  by means of a coinductive formal system, as    
first done in~\cite{BDL24} for synchronous communications without mixed choices. 
Such a coinductive definition enables us to consider global types that contain 
branches where some communication can be  postponed indefinitely. 
To achieve this, it is useful to associate 
a global type  with  the set of communication  labels that may (but not necessarily) decorate its transitions. 
 We call these the capabilities of the global type.

\begin{definition}[Capabilities]\label{cp}~~

\Cline{\cp{\End}=\emptyset\quad\cp{\cml.\G}=\set{\cml}\cup\cp{\G}\quad\cp{\G_1+\G_2}=\cp{\G_1}\cup\cp{\G_2}}   

\end{definition} 

As  usually done in SMPS with  
asynchronous communication, we  introduce an LTS on  
{\em type configurations} of the form $\G\parG\Msg$, i.e. 
parallel compositions of a global type and a queue.  We use 
$\cml.\G \pplus \G'$   as short for either $\cml.\G+\G'$ or $\cml.\G$. 

\begin{definition}[LTS for  Type Configurations]\label{ltsgt}
  The {\em labelled transition system (LTS) for type configurations} is
  specified by the following axiom and 
  rule: 

\Cline{\begin{array}{c}
\NamedRule{\rn{GE}}{}{(\cml.\G\pplus\G')\parG \Msg\SLTS{\cml}\G\parG \cml(\Msg)}{}{}
\\[2mm]
\rn{GI}\quad
\prooftree \G_i\parG \Msg\SLTS{\Lambda}\G'_i\parG \Lambda(\Msg)
\quad \players\Lambda\neq\players{\Lambda_i} 
\quad
\Lambda\in\cp{\G_i}  \quad \forall i\in I
\Justifies
\Sigma_{i\in I} \Lambda_i.\G_i\parG \Msg\SLTS{\Lambda}\Sigma_{i\in I} \Lambda_i.\G'_i\parG  
\Lambda(\Msg)
\endprooftree
\end{array}}
\end{definition}

\noindent
Notice that the implicit condition on Axiom \rn{GE}   and Rule \rn{GI}  to be applied is ``$\cml(\Msg)$ is defined''.

The condition $\cml\in\cp{\G_i}$ in Rule \rn{GI} is needed  because  of the coinductive nature of 
such a rule. 
 In fact, without such a condition,  
   we could get $\G\parN\emptyset\SLTS{\pp\pq!\msg{\lambda}}\G\parN\mes\pp{\msg\lambda}\pq$
  for $\G=\pr\ps!\msg{\lambda'}.\ps\pr?.\msg{\lambda'}.\G$ by means of the 
 following infinite derivation:

\Cline{
\mathcal{D}= \quad
                       \prooftree
                              \prooftree
                              \mathcal{D}
                              \Justifies
                              \ps\pr?.\msg{\lambda'}.\G\parN\emptyset  \SLTS{\pp\pq!\msg\lambda} \ps\pr?.\msg{\lambda'}.\G\parN\mes\pp{\msg\lambda}\pq
                              \using \rn{GI}
                              \endprooftree
                       \Justifies
                       \G\parN\emptyset\SLTS{\pp\pq!\msg\lambda}\G\parN\mes\pp{\msg\lambda}\pq
                       \using \rn{GI}
                       \endprooftree
            }

\bigskip
    
            We provide now an extension of \cite[Example 2]{MMSZ21}, where we use mixed choice in order
to manage the case when the server with two workers decides to interrupt the interactions. 
\begin{example}[Client Server Workers]
\label{ex:csw1w2}
{\em
A client $\pc$ can keep on sending task requests ($\msg{req}$) to a server $\ps$ until $\ps$ is not able to serve $\pc$ anymore. In the latter case $\ps$ sends a $\msg{halt}$ \tg\ to $\pc$ and both stop.
The server $\ps$, on receiving a task request from $\pc$, decides whether  $\pc$
can be better served by 
the worker $\pw_1$ or by the worker $\pw_2$ and then sends the task request accordingly to one of them.
Once the task is performed, the result ($\msg{res}$) is sent to $\pc$ by either $\pw_1$ or $\pw_2$.
When $\ps$ decides to stop, 
$\ps$ informs $\pw_1$ (via  the   \tg\ $\msg{last}$) when the request the latter is going to receive will be the last one sent by $\ps$ before stopping. Such information is also communicated by $\pw_1$ to $\pc$ by using $\msg{resL}$ instead of $\msg{res}$. 
 After the last request to $\pw_1$, participant ${\ps}$ sends also a \tg\ $\msg{halt}$ to $\pw_2$ in order to make the latter terminate. Processes implementing this protocol are:  
 
\Cline{
\PP_{\pc} = \ps!\msg{req}.(\pw_1?\msg{res}.\PP_{\pc} + \pw_2?\msg{res}.\PP_{\pc}+\pw_1?\msg{resL} )
}

\Cline{
\PP_{\ps} =  \pc?\msg{req}.( \pw_1!\msg{req}.\PP_{\ps} +
                                                                        \pw_2!\msg{req}.\PP_{\ps})+\pw_1!\msg{last}.\pc?\msg{req}.\pw_1!\msg{req}.\pw_2!\msg{halt} 
}

\Cline{
\PP_{\pw_1} =  \ps?\msg{req}.\pc!\msg{res}.\PP_{\pw_1} + \ps?\msg{last}.\ps?\msg{req}.\pc!\msg{resL}
\qquad\qquad
\PP_{\pw_2} =  \ps?\msg{req}.\pc!\msg{res}.\PP_{\pw_2} + \ps?\msg{halt}
}

\noindent
 The session $\pP\pc{\PP_\pc}\parN\pP\ps{\PP_\ps}\parN\pP{\pw_1}{\PP_{\pw_1}}\parN\pP{\pw_2}{\PP_{\pw_2}}\parN\emptyset$ can be typed by 

\Cline{\G=\pc\ps!\msg{req}.(\ps\pc?\msg{req}.(\ps\pw_1!\msg{req}.\pw_1\ps?\msg{req}.\pw_1\pc!\msg{res}.\pc\pw_1?\msg{res}.\G+
\ps\pw_2!\msg{req}.\pw_2\ps?\msg{req}.\pw_2\pc!\msg{res}.\pc\pw_2?\msg{res}.\G)+\G')
}

\noindent
where $\G'=\ps\pw_1!\msg{last}.\pw_1\ps?\msg{last}.\ps\pc?\msg{req}.\ps\pw_1!\msg{req}.\pw_1\ps?\msg{req}.\pw_1\pc!\msg{resL}.\pc\pw_1?\msg{resL}.\ps\pw_2!\msg{halt}. \pw_2\ps?\msg{halt} 
$.
} 
\finex
\end{example}

\smallskip
In the following we provide  another a simple example exploiting the expressiveness of the mixed choice.  Its implementation as a network can be typed also by global types that do  
not correspond to the tree of all the possible actions.

\begin{example}[Time Out]
{\em
Participant $\pp$ wishes to send a \tg\  $\msg{\tge}$ to $\ps$. The value carried by $\msg{\tge}$
is processed by the participants $\pq$ and $\pr$ and then sent by 
the latter to  $\pp$ with the \tg\ $\msg{v}$. In case $\pp$ does not receive the \tg\ 
from $\pr$ before a certain amount of time, $\pp$ sends the \tg\ $\msg{\tge}$ to $\ps$ using  
instead a default value. 
In such a case, when the \tg\ $\msg{v}$ is eventually received by $\pp$, is discarded.
From the point of view of communications, the previous behaviour is implemented by the 
session with the network

\Cline{
\Nt = \pP\pp{\pr?\msg{v}.\ps!\msg{\tge} + \ps!\msg{\tge}.\pr?\msg{v}} \parN
\pP\ps{\pp?\msg{\tge}} \parN
\pP\pr{\pq?\msg{v}.\pp!\msg{v}} \parN
\pP\pq{\pr!\msg{v}} 
}

\noindent
 and the empty queue. 
For this session, it is possible to get a derivation for

\Cline{\G\vdash \Nt \parN \emptyset}

\noindent
where 

 \Cline{\G = \pq\pr!\msg{v}.\pr\pq?\msg{v}.\pr\pp!\msg{v}.(\pp\pr?\msg{v}.\pp\ps!\msg{\tge}.\ps\pp?\msg{\tge}+\pp\ps!\msg{\tge}.\ps\pp?\msg{\tge}.\pp\pr?\msg{v})}
 
 \noindent
 or 
 
 \Cline{\G= \pp\ps!\msg{\tge}.\pq\pr!\msg{v}.\ps\pp?\msg{\tge}.\pr\pq?\msg{v}.\pr\pp!\msg{v}.\pp\pr?\msg{v}
+\pq\pr!\msg{v}.\pr\pq?\msg{v}.\pr\pp!\msg{v}.\pp\pr?\msg{v}.\pp\ps!\msg{\tge}.\ps\pp?\msg{\tge}
}
 
 \noindent
 or $\G$ is a type obtained from the given ones by permuting communication labels with different players, but preserving the order between an output label and the corresponding input label.
}
\finex
\end{example}

\section{Properties}\label{pr}

It is essential that session reductions preserve satisfaction of non-involved participants.

\begin{lemma}[Satisfaction Preservation]\label{sp}
If $\Nt\parN \Msg\SLTS{\Lambda}\Nt'\parN \Msg'$  with $\plays{\Lambda}\neq\set\pp$ and $\pp$ is satisfied in $\Nt\parN \Msg$, then $\pp$ is satisfied in $\Nt'\parN \Msg'$.
\end{lemma}

 If a global type can type a session, the players of the global type and of the network 
coincide: this can easily be shown by coinduction on type derivations.

\begin{lemma}[Players]
\label{l:p1}
If $\ty\G{\Nt\parN \Msg}$, then  $\players\G=\players\Nt$. 
\end{lemma}

Subject Reduction ensures that  session reductions can be   
mimicked by the reductions of the type configurations with the global types which type them and the same queues. The proof  can be carried on by coinduction on the type derivation and by cases on the last applied axiom/rule.

\begin{theorem}[Subject Reduction]\label{sr}
If $\ty\G\Nt\parN \Msg$ and $\Nt\parN \Msg\SLTS{\Lambda}\Nt'\parN \cml(\Msg)$, then $\ty{\G'}{\Nt'\parN   \cml(\Msg)}$ and $\G\parG \Msg\SLTS{\Lambda}\G'\parG \cml(\Msg)$ for some $\G'$.
\end{theorem}

Session Fidelity is a sort of reverse of Subject Reduction:  reductions of a  type configuration 
can be mimicked by reductions of the sessions typed by the global type. 
This can be proved  
by coinduction on the type derivation 
and by cases on the last applied axiom/rule. 

\begin{theorem}[Session Fidelity]\label{sf}
If $\ty\G\Nt\parN \Msg$ and $\G\parG \Msg\SLTS{\Lambda}\G'\parG \cml(\Msg)$, then $\ty{\G'}{\Nt'\parN \cml(\Msg)}$ and \linebreak $\Nt\parN \Msg\SLTS{\Lambda}\Nt'\parN \cml(\Msg)$ for some $\Nt'$.
\end{theorem}

Toward establishing the property that typable sessions are lock free, the following lemma is handy. 
We  first  define $\plays{\epsilon}=\emptyset$ and $\plays{\Lambda\cdot\sigma}=\plays{\Lambda}\cup\plays{\sigma}$. 

\begin{lemma}[Reduction of Global Types] \label{lem:G-participants}
  If  $\ty\G{\Nt\parN\Msg}$ and  $\pp \in \plays{\G}$, then there are  $\sigma$,  $\cml$, $\G'$ and $\Msg'$ such that  
 
 \Cline{
 \G\parG \Msg  \SLTS{\;\sigma\cdot\cml\;}  \G'\parG \Msg'
 } 
 
 \noindent
 where $\pp \not\in\plays{\sigma} $   and $\pp \in \plays{\cml}$.
\end{lemma}

 We are now in a position to demonstrate that typable sessions are both lock and orphan-message free.

\begin{theorem}[Lock Freedom]
 If \hspace{1pt}$\Nh$ is typable, then $\Nh$ is lock free.
\end{theorem}
\begin{proof}
Let $\ty\G\Nh$.  Following Definition~\ref{d:lf}, in order to prove Lock Freedom for $\Nh$, 
let $\Nh \SLTS{\sigma}\Nt\parN \Msg$ 
and let $\pp\in\plays{\Nt}$.
By  Subject Reduction (Theorem \ref{sr}) we get   $\ty{\G'}{\Nt\parN \Msg}$  for some $\G'$.
We can now recur to Lemma~\ref{l:p1} and get $\pp \in  \plays{\G'}$.
From the fact that $\pp \in \plays{\G'}$ and by Lemma \ref{lem:G-participants},
it follows that $\G' \parG\Msg \SLTS{\sigma'\cdot\cml\;}  \G''\parG\Msg'$ for some  $\sigma'$ and  $\cml$ with  $\pp \not\in \plays{\sigma'}$ and $\pp \in \plays{\cml}$. Now the thesis follows
by Session Fidelity (Theorem \ref{sf}). 
\qed
\end{proof}

\begin{theorem}[Orphan-message Freedom] 
If \hspace{1pt}$\Nh$ is typable, then $\Nh$ is orphan-message free.
\end{theorem}
\begin{proof}
Let $\Nh\SLTS{\sigma}\Nt'\parN\Msg'$    with  $\Nt'$ final. 
By Subject Reduction (Theorem \ref{sr}) we get $\ty{\G'}{\Nt'\parN \Msg'}$
for some $\G'$. Since $\Nt'$ is final, $\Nt'\parN \Msg'$ can be typed only by means of 
Axiom \rn{End}. Hence we have necessarily that $\G'=\End$ and $\Msg'=\emptyset$.
\qed
\end{proof}

\medskip

We now formalise the property  that guarantees  any message   at the top of a session's queue 
can be eventually consumed \cite{LY19}.  This also implies that all messages in the session's queue 
can be eventually consumed. We write $\Lambda\in\sigma$ if $\sigma=\sigma'\cdot\Lambda\cdot\sigma''$ for some $\sigma'$, $\sigma''$.

\begin{definition}[Eventual Reception]
A session $\Nt\parG\Msg$  satisfies  
the {\em Eventual Reception} property if \linebreak
$\Nt\parG \Msg  
\SLTS{\sigma}\Nt'\parG\Msg'$, with 
$\Msg'\equiv\mes\pp{\msg{\lambda}}\pq\cdot\Msg''$, implies  that there exists $\sigma'$ 
 such that
$\Nt'\parG\Msg'\SLTS{\sigma'}\SLTS{\pq\pp?\msg{\lambda}}$,  where 
$\pq\pp?\msg{\lambda}\not\in\sigma'$. 
\end{definition}

We let $\mse,\mse'\ldots$ range over the set of messages and $\GG$  range  
over finite sets of global types. We define the {\em weight} of $\mse$ in $\G$, notation $\w\G\mse$,  by coinduction on $\G$.

\begin{definition}[Weight]
 Let $\mse=\mes\pp{\msg{\lambda}}\pq$. We define: 
\begin{enumerate}[i)]
\item~
\Cline{\begin{array}{l}\waux\End\mse\GG=\infty\\ \waux{\cml.\G}\mse\GG=\begin{cases}
  0    & \text{if }\cml=\pq\pp?\msg{\lambda} \\
  \infty    & \text{if }\cml=\pq\pp?\msg{\lambda'} \text{ or }  \cml.\G\in \GG \\
 1+\waux\G\mse{\GG\cup\set{\cml.\G}}     & \text{otherwise}
\end{cases}\\
~\\[-10pt]
\waux{\Sigma_{i\in I} \Lambda_i.\G_i}\mse\GG=min_{i\in I} \waux{ \Lambda_i.\G_i}\mse\GG
\end{array}}
\item
 $\w\G\mse = \waux\G\mse\emptyset$. 
\end{enumerate}
\end{definition}

Notice that $\waux\G\mse\GG$  (and hence $\w\G\mse$)  is well defined since global types are regular.  For example if $\G=\pp\pq!\msg{\lambda}.\pq\pp?\msg{\lambda}.\G + \pp\pr?\msg{\lambda}$, then 
$\w\G{\mes\pp{\msg{\lambda}}\pq}=1$ and $\w\G{\mes\pr{\msg{\lambda}}\pp}=0$. 

\begin{definition}[Soundness]
A type configuration $\G\parG\Msg$ is {\em sound} if $\w\G\mse$ is finite for all \ms s $\mse$ which occur in $\Msg$.
\end{definition}

Let \rn{TCommS} be the typing rule obtained from \rn{TComm} by adding the condition ``$\Sigma_{i\in I} \Lambda_i.\G_i\parG\Msg$ is sound''. We denote by $\vdash_S$ the derivability 
 relation  in the  so  obtained type assignment system.

\begin{theorem}[Eventual Reception]\label{er}
If $\G\vdash_S \Nt\parN\Msg$, then $\Nt\parN\Msg$ satisfies the Eventual Reception property.
\end{theorem}

For example, the client/server session of Example~\ref{ex1} satisfies the Eventual Reception property, since we can derive for this session the type $\G$ defined after Definition~\ref{def:type-system} in the type system $\vdash_S$.

\section{Concluding Remarks, Related and Future Works}\label{rfw}
This paper is a companion to~\cite{BD25}, in which we investigated the use of mixed choice
for synchronous communication in the  setting of SMPS~\cite{DGD22,BBD25a}. To guarantee the properties implied by typability, sessions in~\cite{BD25} are considered implicitly composed 
by modules whose participants interact freely via unrestricted mixed  choice.
In contrast, the inter-module communication can be more easily controlled by 
 allowing a restricted form of mixed choice whereby both input and output actions can involve a single participant only. 
This approach does not discard any session a priori. 
We argue that the notion of coherence of communication label sets is abstractly less demanding than session modularisation and can lead to a simpler type system for synchronous mixed choice multiparty sessions
than that described in~\cite{BD25}.
One relevant feature of our type system is that it is conservative, since for processes without mixed choice it coincides with the coinductive version of the type system in~\cite{BD24} when delegation is disregarded. 
 Clearly with asynchronous communication less sessions are stuck than with synchronous communication, the standard example being the network $\pP\pp{\pq!\msg{\lambda}.\pq?\msg{\lambda'}}\parN\pP\pq{\pp!\msg{\lambda'}.\pp?\msg{\lambda}}$. The session obtained adding the empty queue to this network is typable in the current type system, while this network cannot be typed in the type system of~\cite{BD25}.  

In~\cite{CFJ26} processes without mixed choices are directly verified in a compositional way toward global types without local types and projections. It would be challenging to adapt this approach to processes with mixed choices.

 Mixed choice 
 has been also investigated in~\cite{CMV22,PY24b} 
for binary session types,  whereas~\cite{PY24} considers mixed choice in a MPST  setting~\cite{HYC08,Honda2016},  
following the approach of~\cite{ScalasY19} where
global types are not taken into account.  
Although the main concern  of~\cite{PY24} is the expressivity of multiparty calculi 
(according to the full range of possible restrictions of mixed choice), 
type systems assigning local types to processes are provided, 
where various predicates on contexts of local types are investigated.  
In~\cite{PBK23,PBMK24} asynchronous binary sessions with timeout and mixed choice are enriched with 
a semantics guaranteeing Progress and a type system satisfying Subject Reduction. 
 The authors of \cite{MB0B26} developed a
 mechanised higher-order separation logic for proving functional correctness of higher-order imperative programs with mixed-choice multiparty message passing and shared memory.

As previously mentioned, there is a precise correspondence between the participants in our multiparty sessions and Communicating Finite State Machines (CFSMs) ~\cite{bz83, cf05}. 
This opens up the possibility of exploiting the structured SMPS framework to investigate a typically unstructured formalism like CFSM. In particular, the use of mixed choice whose uncontrolled appliance brings with it subtly harmful features, as was exposed decades ago~\cite{GMY80}.
To this end, a structured choreographic framework based on the concept of projection was recently proposed in~\cite{SD25}. In this framework, local views are represented as CFSMs and global views as PSTs (Protocol State Machines). The latter is an automata-based formalism that uses input/output actions as edge labels. Interestingly, PSTs can be shown to be equivalent to our global types, except for the presence of final states in PSTs where the notion of final is as in standard automata for regular language recognition.
This immediately prompts an investigation into the convergence of the SMPS framework with the APM (Automata-based Multiparty Protocols) machinery, which is based on PTS and CFSM. This machinery was proposed in~\cite{SD25} as a foundation for top-down protocol design.
In fact, such a convergence was already hinted at in~\cite{SD25}, where the authors draw attention to the compatibility of their APM machinery with scenarios where an algorithm infers a PSM from CSMs.
Some restrictions are devised in~\cite{SD25} in order to guarantee relevant properties.
One of such restrictions is related, from our setting viewpoint, to the boundedness of 
messages in the queue with the same sender-receiver pair. 
For what concerns our typing system, we are definitely interested in investigating the conditions that ensure the decidability of the typing relation.  This could involve extending the approach used in \cite{BDL23} and \cite{BD24}, where an invariant requirement ensures decidability for
an inductive version of a coinductive system (in the style of~\cite[Section  21.9]{pier02}).
Such invariance implying that any output in a cycle must have
 a corresponding input within that same cycle.  

We envisage other avenues of investigation stemming from the present paper's approach to verifying properties for asynchronous multiparty sessions and, consequently, for CFSMs.
In particular taking inspiration from results like those in ~\cite{DY13} and, more specifically, in ~\cite{LY19}, where conditions ensuring relevant communication properties for CFSMs were devised based on the syntactic properties of projectable global types.
In our present setting, we plan to explore conditions on CFSMs building on
the notion of coherence of communication label sets.
 We also claim that disregarding Item (\ref{cscl2}) from  Definition~\ref{cscl} 
would result in a stronger form of Lock Freedom. This comes at the cost of restricting the set of typable terms. 
We are currently investigating a type system using such a restricted
definition of coherence in a setting where  mixed choices are used as an abstraction for 
timeouts and alternative computation paths. 

\subsection*{Acknowledgements} We wish to warmly 
thank the  anonymous reviewers 
for their thoughtful and helpful comments.

\bibliographystyle{eptcs}
\bibliography{session}

%
%
%

\end{document}